\documentclass[10pt,aps,prl,twocolumn]{revtex4-2}

\usepackage{mathrsfs}
\usepackage{bbm}
\usepackage{amsmath}
\usepackage{amssymb}
\usepackage{graphicx}
\usepackage{mathtools}
\usepackage{amsthm}

\makeatletter

\usepackage{hyperref}
\usepackage{color}
\definecolor{myblue}{RGB}{0,50,200}
\hypersetup{
colorlinks,
citecolor=myblue,
linkcolor=myblue,
urlcolor=myblue
}
\allowdisplaybreaks

\newcommand{\bms}{\mathbbmss}
\newcommand{\mca}{\mathcal}
\newcommand{\mbb}{\mathbb}

\newcommand{\avg}[1]{\langle #1\rangle}

\newcommand{\bra}[1]{\left( #1 \right)}

\newcommand{\pp}{\partial}

\newcommand{\Tr}[1]{{\rm tr}\left\{#1\right\}}
\newcommand{\Br}[1]{\langle #1|}
\newcommand{\Kt}[1]{|#1\rangle}
\newcommand{\Mat}[1]{|#1\rangle\langle #1|}
\newcommand{\dMat}[2]{|#1\rangle\langle #2|}

\newcommand{\avgObs}[2]{\langle #1|#2|#1\rangle}
\newcommand{\davgObs}[3]{\langle #1|#2|#3\rangle}

\begin{document}
\title{Lower Bound on Irreversibility in Thermal Relaxation of Open Quantum Systems}

\author{Tan Van Vu}
\email{tan@biom.t.u-tokyo.ac.jp}
\affiliation{Department of Information and Communication Engineering, Graduate
School of Information Science and Technology, The University of Tokyo,
Tokyo 113-8656, Japan}

\author{Yoshihiko Hasegawa}
\email{hasegawa@biom.t.u-tokyo.ac.jp}
\affiliation{Department of Information and Communication Engineering, Graduate
School of Information Science and Technology, The University of Tokyo,
Tokyo 113-8656, Japan}

\date{\today}

\begin{abstract}
We consider the thermal relaxation process of a quantum system attached to single or multiple reservoirs.
Quantifying the degree of irreversibility by entropy production, we prove that the irreversibility of the thermal relaxation is lower-bounded by a relative entropy between the unitarily-evolved state and the final state.
The bound characterizes the state discrepancy induced by the non-unitary dynamics, and thus reflects the dissipative nature of irreversibility.
Intriguingly, the bound can be evaluated solely in terms of the initial and final states and the system Hamiltonian, thereby providing a feasible way to estimate entropy production without prior knowledge of the underlying coupling structure.
This finding refines the second law of thermodynamics and reveals a universal feature of thermal relaxation processes.
\end{abstract}

\maketitle

\emph{Introduction.}---The last two decades have witnessed substantial progress in the thermodynamics of nonequilibrium systems subject to significant fluctuations.
Various properties of small systems have been elucidated with the advent of comprehensive frameworks such as stochastic thermodynamics \cite{Sekimoto.2010,Seifert.2012.RPP} and quantum thermodynamics \cite{Vinjanampathy.2016.CP,Goold.2016.JPA,Deffner.2019}.
One of the prominent universal relations is the celebrated fluctuation theorem \cite{Evans.1993.PRL,Gallavotti.1995.PRL,Crooks.1999.PRE,Jarzynski.2000.JSP,Esposito.2009.RMP,Campisi.2011.RMP}, from which the second law of thermodynamics and the fluctuation-dissipation theorem can be immediately derived \cite{Gallavotti.1996.PRL,Andrieux.2007.JSM,Saito.2008.PRB}.
Beyond the fluctuation theorem, much recent attention has been focused on thermodynamic uncertainty relations \cite{Barato.2015.PRL,Gingrich.2016.PRL,Horowitz.2017.PRE,Proesmans.2017.EPL,Brandner.2018.PRL,Hasegawa.2019.PRE,Vu.2019.PRE.UnderdampedTUR,Hasegawa.2019.PRL,Timpanaro.2019.PRL,Guarnieri.2019.PRR,Carollo.2019.PRL,Vo.2020.PRE,Dechant.2020.PNAS,Hasegawa.2020.PRL,Vu.2020.PRR,Koyuk.2020.PRL,Miller.2021.PRL.TUR,Hasegawa.2021.PRL,Sacchi.2021.PRE,Horowitz.2020.NP}, speed limits \cite{Mandelstam.1945.JP,Margolus.1998.PD,Campo.2013.PRL,Pires.2016.PRX,Deffner.2017.JPA,Okuyama.2018.PRL,Shiraishi.2018.PRL,Shanahan.2018.PRL,Funo.2019.NJP}, and refinements of the second law \cite{GomezMarin.2008.PRE,Vaikuntanathan.2009.EPL,Deffner.2010.PRL,Aurell.2011.PRL,Alhambra.2017.PRA,Mancino.2018.PRL,Dechant.2019.arxiv,Shiraishi.2019.PRL,Abiuso.2020.E,Buscemi.2020.PRA,Campisi.2021.PRE,Vu.2021.PRL}.
These findings are not only theoretically important but also provide powerful tools for thermodynamic inference, for instance, in the estimation of free energy \cite{Gore.2003.PNAS} and dissipation \cite{Li.2019.NC,Manikandan.2020.PRL,Vu.2020.PRE,Otsubo.2020.PRE}.

Central to most established relations is thermodynamic irreversibility, which is quantified by irreversible entropy production.
The positivity of entropy production is universally captured by the second law of thermodynamics, which imposes fundamental limits on the computational cost via Landauer's principle \cite{Landauer.1961.JRD,Reeb.2014.NJP,Goold.2015.PRL,Timpanaro.2020.PRL} and on the performance of physical and biological systems such as heat engines \cite{Shiraishi.2016.PRL,Pietzonka.2018.PRL} and molecular motors \cite{Pietzonka.2016.JSM}.
The importance of entropy production has triggered intense research to formulate and investigate its properties \cite{Saito.2016.EPL,Neri.2017.PRX,Pigolotti.2017.PRL,Manzano.2019.PRL,Falasco.2020.PRL}, across from classical to quantum (see Ref.~\cite{Landi.2021.RMP} for a review).
Nonetheless, with restriction to a specific class of nonequilibrium processes, rich features of thermodynamic irreversibility may be found.
One of the interesting classes is thermal relaxation, which is ubiquitous in nature and plays crucial roles in condensed matter \cite{Dattagupta.2012}, heat engines \cite{Benenti.2017.PR}, and quantum state preparation.
Any system coupled to thermal reservoirs unavoidably exchanges energy with the surrounding environment and relaxes to a stationary state.
Of note, Ref.~\cite{Shiraishi.2019.PRL} proved that the entropy production during relaxation of classical Markovian processes is bounded from below by the classical relative entropy between the initial and final distributions.
For relaxation to equilibrium in open quantum systems, it has been shown that the entropy production is lower-bounded by a lag between states in terms of a time-reversed map \cite{Alhambra.2017.PRA} and a geometrical distance on the Riemannian manifold \cite{Vu.2021.PRL}.

In this Letter, we deepen our understanding of thermodynamic irreversibility in thermal relaxation processes of open quantum systems.
Specifically, we derive a fundamental bound on irreversibility for systems that are in contact with thermal reservoirs and described by the Lindblad master equation.
We prove that the irreversible entropy production during relaxation is lower-bounded by a relative entropy between the final states obtained with the unitary dynamics and the original dynamics [cf. Eq.~\eqref{eq:main.result.1}].
Lindblad dynamics comprise unitary and non-unitary parts; therefore, the lower bound quantifies the state discrepancy induced by the dissipative non-unitary dynamics, which intuitively reflects the nature of thermodynamic irreversibility (shown in Fig.~\ref{fig:illustration}).
Remarkably, the derived bound is saturable, experimentally accessible, and stronger than the conventional second law of thermodynamics.
Furthermore, the bound provides a feasible way to estimate irreversible entropy production with the help of the quantum state tomography technique.
We show that a tighter Landauer bound can be obtained as a corollary and demonstrate our result in an autonomous thermal machine.

\emph{Main result.}---We consider the thermal relaxation process of a Markovian open quantum system.
The system can be simultaneously coupled to multiple thermal reservoirs at different temperatures.
The dynamics of the density matrix $\rho(t)$ are governed by the Lindblad master equation \cite{Lindblad.1976.CMP}:
\begin{equation}
\pp_t\rho(t)=\mca{L}[\rho(t)]\coloneqq -i[H,\rho(t)]+\sum_k\mca{D}_k[\rho(t)],\label{eq:org.Lindblad.eq}
\end{equation}
where $H$ is the time-independent Hamiltonian and the dissipator is given by $\mca{D}_k[\circ]\coloneqq L_k(\circ)L_k^\dagger-\frac{1}{2}\{L_k^\dagger L_k,(\circ)\}$.
The jump operators come in pairs $(L_k,L_{k'})$ with energy changes $(\omega_k,\omega_{k'})$, which satisfy $[L_k,H]=\omega_kL_k$ and $\omega_k=-\omega_{k'}$.
This condition implies that the jump operators account for transitions between different energy eigenbasis with the same energy change \cite{Manzano.2018.PRX,Funo.2019.NJP}. 
Note that $[\circ,\star]$ and $\{\circ,\star\}$ are, respectively, the commutator and the anticommutator of two operators, and the Planck constant and the Boltzmann constant are both set to unity throughout this Letter, $\hbar=k_{\rm B}=1$.
We assume the local detailed balance condition $L_k=e^{\Delta s_{\rm env}^k/2}L_{k'}^\dagger$, which is fulfilled in most cases of physical interest \cite{Manzano.2019.PRL}.
Here $\Delta s_{\rm env}^{k}$ is the change in the environment entropy due to a jump of type $k$.
After a sufficiently long time, the system reaches a stationary state that may no longer be a Gibbs state when multiple reservoirs are attached.
The degree of irreversibility of the relaxation process during time $\tau$ can be quantified by the irreversible entropy production $\Sigma_\tau$, defined as
\begin{equation}
\Sigma_\tau\coloneqq\Delta S_{\rm sys}+\Delta S_{\rm env},\label{eq:irr.ent.prod}
\end{equation}
where $\Delta S_{\rm sys}\coloneqq\Tr{\rho(0)\ln\rho(0)}-\Tr{\rho(\tau)\ln\rho(\tau)}$ is the change in the system entropy (characterized by the von Neumann entropy) and $\Delta S_{\rm env}$ corresponds to the entropy change of the environment, given by \cite{Manzano.2018.PRX}
\begin{equation}
\Delta S_{\rm env}\coloneqq\int_0^\tau\sum_k\Tr{L_k\rho(t)L_k^\dagger}\Delta s_{\rm env}^kdt.
\end{equation}
Within this definition, it can be proved that the entropy production is always nonnegative, $\Sigma_\tau\ge 0$ \cite{Trushechkin.2019.JMS}.

\begin{figure}[t]
\centering
\includegraphics[width=1.0\linewidth]{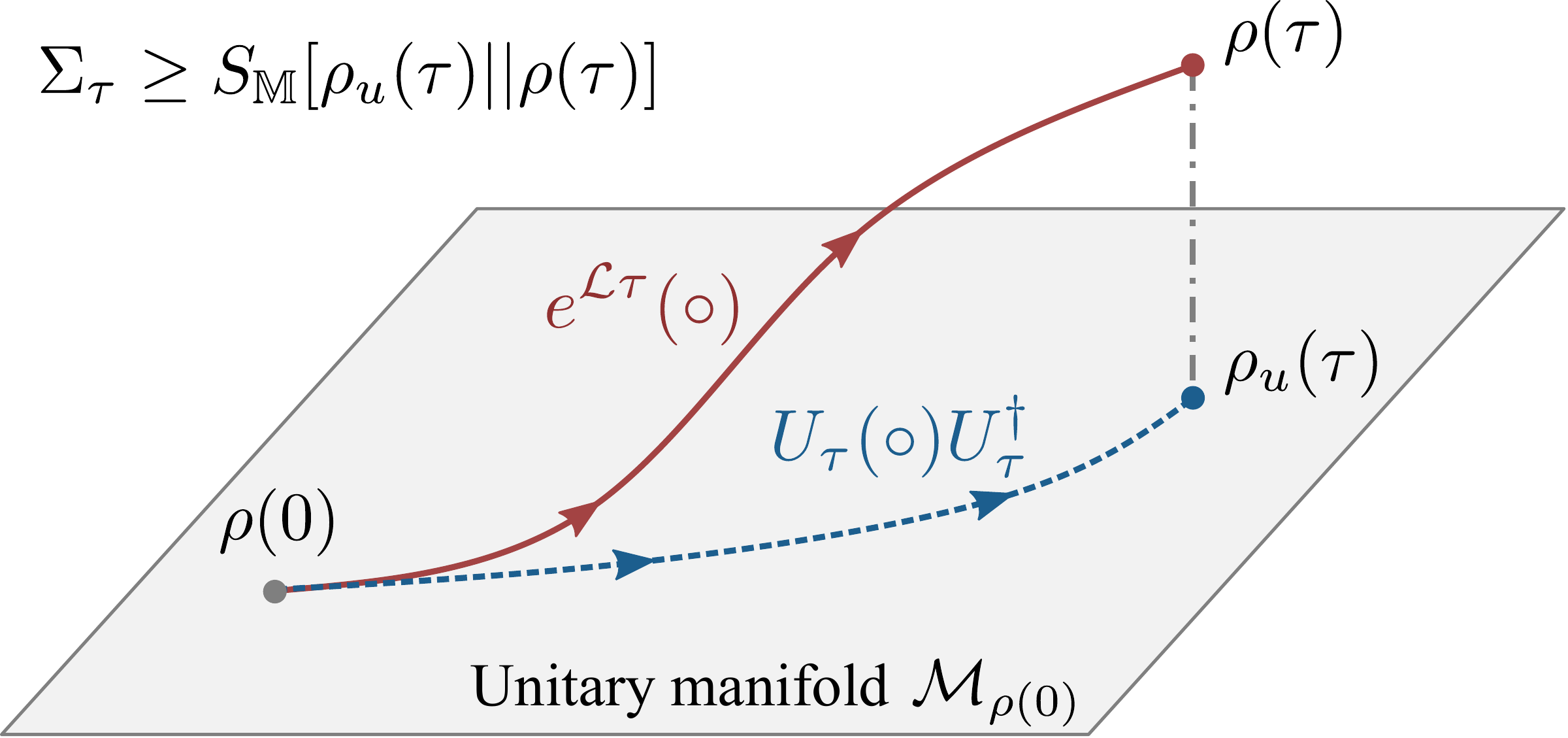}
\protect\caption{Geometrical illustration of the main result. The manifold of density matrices that can be generated from $\rho(0)$ via nondissipative unitary transforms is denoted by $\mca{M}_{\rho(0)}$. The time evolution of the density matrix under the Lindblad dynamics and the unitary dynamics is described by the solid line and the dashed line, respectively. The irreversible entropy production $\Sigma_\tau$ is bounded from below by the information-theoretical distance $S_{\mbb{M}}[\rho_u(\tau)||\rho(\tau)]$ --- a relative entropy between the unitarily-evolved state $\rho_u(\tau)\coloneqq U_\tau\rho(0)U_\tau^\dagger$ and the final state $\rho(\tau)$.}\label{fig:illustration}
\end{figure}

We explain our main result under the given setup.
The proof is provided at the end of the Letter.
Our main result is a lower bound on $\Sigma_\tau$ in terms of a relative entropy between the initial and final states,
\begin{equation}
\Sigma_\tau\ge S_{\mbb{M}}[U_\tau\rho(0)U_\tau^\dagger||\rho(\tau)],\label{eq:main.result.1}
\end{equation}
where $U_t\coloneqq e^{-iHt}$ is the unitary operator and $S_{\mbb{M}}(\rho||\sigma)\coloneqq S(\rho||\sum_n\Pi_n\sigma\Pi_n)$ is the projectively measured relative entropy between $\rho$ and $\sigma$ with the eigenbasis $\{\Pi_n\}$ of $\rho$.
Here, $S(\rho||\sigma)\coloneqq \Tr{\rho(\ln\rho-\ln\sigma)}\ge 0$ is the quantum relative entropy between states $\rho$ and $\sigma$.
The inequality \eqref{eq:main.result.1} indicates that the entropy production is bounded from below by an information-theoretical distance between the initial and final states, which strengthens the Clausius inequality in the conventional second law of thermodynamics for thermal relaxation processes.

Some remarks regarding the main result are in order.
First, the bound is geometrically and intuitively understandable.
The system state is governed by Lindblad dynamics, which consist of a non-dissipative unitary part and a dissipative non-unitary part.
The lower bound is the distance between the unitarily-evolved state and the final state; therefore, it quantifies how far the system is driven by non-unitary dynamics, and thereby intuitively reflects the nature of entropy production, namely a dissipative term.
When the system is uncoupled to the reservoirs and governed by unitary dynamics, both the entropy production and relative entropy vanish, and the derived relation becomes a trivial equality.
Second, the bound is tight and can be saturated, for example, in the long-time regime; consequently, it can be applied to thermodynamic inference.
In particular, given the initial and final states and the system Hamiltonian, a lower bound of entropy production can be estimated without requisite prior knowledge of the underlying dynamics.
Third, the bound can be interpreted as a quantum speed limit, $\tau\ge S_{\mbb{M}}[U_\tau\rho(0)U_\tau^\dagger||\rho(\tau)]/\overline{\Sigma}$, where $\overline{\Sigma}\coloneqq\Sigma_\tau/\tau$ is the average entropy production rate.
An important implication of this speed limit is that a fast state-transformation requires a high dissipation rate \cite{Shiraishi.2018.PRL,Vu.2021.PRL}.
Last, in the classical limit (e.g., when the initial density matrix has no coherence in the energy eigenbasis of the Hamiltonian), the lower bound reduces exactly to the classical relative entropy between the initial and final distributions; therefore, our result generalizes the classical bound reported in Ref.~\cite{Shiraishi.2019.PRL} to the case of multiple reservoirs.
Note also that the relation \eqref{eq:main.result.1} is valid even for systems with broken time-reversal symmetry, such as electronic systems with a Peierls phase.

A comparison to the existing bounds on entropy production is provided. Some strong bounds were derived for finite-dimension \cite{Reeb.2014.NJP} and zero-temperature \cite{Timpanaro.2020.PRL} environments. Nonetheless, these bounds reduce to the conventional second law for either infinitely large or high-temperature environments.
Another information-theoretical bound \cite{Alhambra.2017.PRA} was derived for thermal relaxation processes; however, this bound is only applicable to the single-reservoir case.
Our bound thus plays an important role in the regimes in which the existing bounds either become trivial or are inapplicable.

\begin{figure*}[!]
\centering
\includegraphics[width=1.0\linewidth]{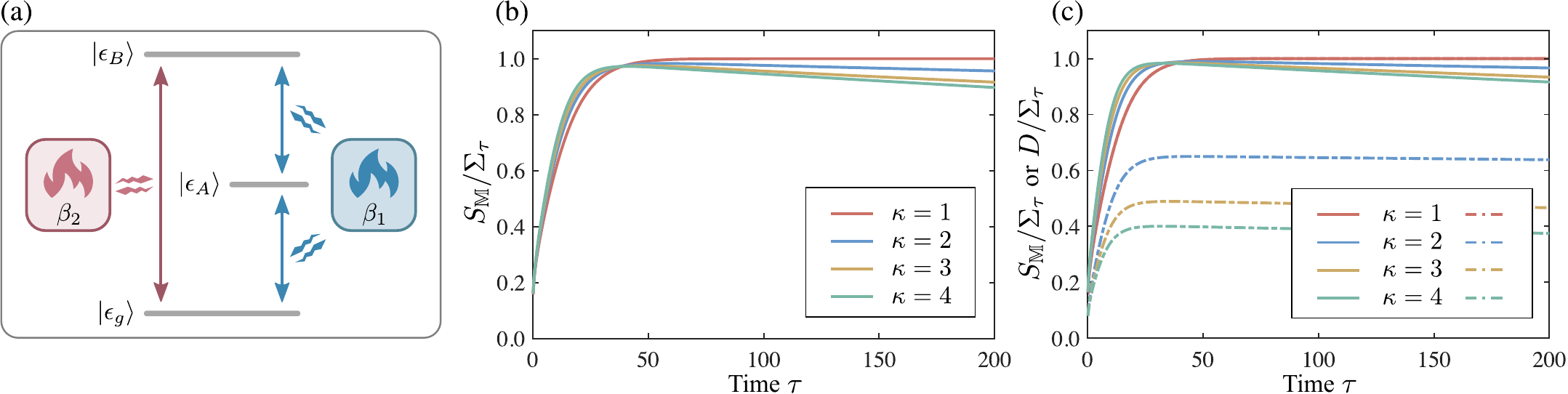}
\protect\caption{Numerical illustration of the main result. (a) Schematic diagram of the two-reservoir machine. (b) The ratio $S_{\mbb{M}}[U_\tau\rho(0)U_\tau^\dagger||\rho(\tau)]/\Sigma_\tau$ is plotted as a function of time $\tau$, and each solid line depicts the result obtained with $\kappa$ ranged from $1$ to $4$. Parameters are fixed as $\omega_1=\omega_3=0.2$, $\gamma=0.01$, and $\beta_1=1$. (c) The ratios $S_{\mbb{M}}[U_\tau\rho(0)U_\tau^\dagger||\rho(\tau)]/\Sigma_\tau$ (solid line) and $D[p_n(0)||p_n(\tau)]/\Sigma_\tau$ (dash-dotted line) are plotted with $\kappa$ ranged from $1$ to $4$. Parameters are fixed as $\omega_1=0.2$, $\omega_3=0.1$, $\gamma=0.01$, and $\beta_1=1$.}\label{fig:num.result.thermal.machine}
\end{figure*}

As coupled to a single thermal reservoir at the inverse temperature $\beta$, the system relaxes toward an equilibrium state $\pi\coloneqq e^{-\beta H}/Z$, irrespective of the initial state.
In this case, a Hamiltonian-free lower bound on the entropy production can be obtained \cite{Supp.PhysRev},
\begin{equation}
\Sigma_\tau\ge S_{\mbb{E}}[\rho(0)||\rho(\tau)],\label{eq:main.result.2}
\end{equation}
where $S_{\mbb{E}}(\rho||\sigma)\coloneqq \sum_n a_n\ln(a_n/b_n)\eqqcolon D(a_n||b_n)$ is exactly the classical relative entropy between distributions $\{a_n\}$ and $\{b_n\}$, which are the increasing eigenvalues of $\rho$ and $\sigma$, respectively.
Any unitary transform does not change the magnitude of the eigenvalues of a density matrix but only the eigenbasis.
Therefore, the term $S_{\mbb{E}}$ quantitatively characterizes the state change caused by the non-unitary dynamics.
The lower bound is now dependent only on the initial and final states, which is an inherent feature of thermal relaxation processes.
For generic time-driven systems, it can be proved that there does not exist such a universal metric that bounds irreversible entropy production from below \cite{Supp.PhysRev}.
However, for a specific class of quantum systems, a metric based on the thermodynamic length can be found \cite{Abiuso.2020.E}.

\emph{Tightening Landauer's principle.}---Information and thermodynamics can be intuitively related via Landauer's principle \cite{Landauer.1961.JRD}, which quantifies the minimal heat dissipation associated with erasure of information in a memory.
The inequality in Eq.~\eqref{eq:main.result.1} derives a tighter bound than the conventional Landauer bound $\Delta S_{\rm env}\ge -\Delta S_{\rm sys}$.
Consider the scenario in which a qubit is reset from a maximally mixed state $\mbb{I}/2$ to the ground state $\Mat{0}$ by quenching the Hamiltonian and letting the system relax to the stationary state.
It should be stressed that the final state cannot be exactly $\Mat{0}$ in a finite time.
Let $\delta\in(0,1)$ be the probability that the qubit is in the ground state at time $\tau$, i.e., $\delta=\avgObs{0}{\rho(\tau)}$.
According to Eq.~\eqref{eq:main.result.1}, the minimal heat dissipation is related to the probability of success as
\begin{equation}
\Delta S_{\rm env}\ge -\Delta S_{\rm sys}-\frac{1}{2}\ln[4\delta(1-\delta)].\label{eq:tight.Land.prin}
\end{equation}
Equation \eqref{eq:tight.Land.prin} implies that the higher precision of a memory is, the more heat is dissipated.
As $\delta\to 1$, the term in the right-hand side of Eq.~\eqref{eq:tight.Land.prin} goes to infinity.

\emph{Examples.}---To demonstrate the main result [Eq.~\eqref{eq:main.result.1}], we consider an autonomous thermal machine \cite{Manzano.2018.PRX} with three levels $\{\Kt{\epsilon_g},\Kt{\epsilon_A},\Kt{\epsilon_B}\}$.
Such machines can operate as refrigerators and are also the building blocks for quantum clocks \cite{Erker.2017.PRX,Mitchison.2019.CP}.
The Hamiltonian of the system is $H=\omega_1\Mat{\epsilon_A}+\omega_2\Mat{\epsilon_B}$, where $\omega_1$, $\omega_2$, and $\omega_3\coloneqq\omega_2-\omega_1$ are frequency gaps between $\Kt{\epsilon_g}\leftrightarrow\Kt{\epsilon_A}$, $\Kt{\epsilon_g}\leftrightarrow\Kt{\epsilon_B}$, and $\Kt{\epsilon_A}\leftrightarrow\Kt{\epsilon_B}$, respectively.
The machine is powered by two reservoirs at different inverse temperatures $\beta_1\ge\beta_2$, which mediate transitions between the energy levels [see Fig.~\ref{fig:num.result.thermal.machine}(a)].

First, we examine the case in which $\omega_1=\omega_3$ and the Lindblad equation has four jump operators, $L_1=\sqrt{\eta_1}(\dMat{\epsilon_A}{\epsilon_g}+\dMat{\epsilon_B}{\epsilon_A})$, $L_{1'}=\sqrt{\eta_{1'}}(\dMat{\epsilon_g}{\epsilon_A}+\dMat{\epsilon_A}{\epsilon_B})$, $L_{2}=\sqrt{\eta_2}\dMat{\epsilon_B}{\epsilon_g}$, and $L_{2'}=\sqrt{\eta_{2'}}\dMat{\epsilon_g}{\epsilon_B}$, where $\eta_{1}=\gamma n_1^{\rm th}(\omega_1)$, $\eta_{1'}=\gamma[n_1^{\rm th}(\omega_1)+1]$, $\eta_{2}=\gamma n_2^{\rm th}(\omega_2)$, and $\eta_{2'}=\gamma[n_2^{\rm th}(\omega_2)+1]$.
Here, $n_r^{\rm th}(\omega)\coloneqq(e^{\beta_r\omega}-1)^{-1}$ denotes the Planck distribution and $\gamma$ is the decay rate.
Let $\pi^{\rm ss}$ be the stationary state of the system.
We set $\beta_1=\kappa\beta_2$ with $\kappa\ge 1$, and the initial density matrix is a pure state $\rho(0)=\Mat{\varphi}$, where $\Kt{\varphi}=\sqrt{10-\kappa}\Kt{\epsilon_g}/3+\sqrt{\kappa-1}\Kt{\epsilon_B}/3$.
The magnitude of $\kappa$ characterizes the coherence in the initial state and the nonequilibrium degree of $\pi^{\rm ss}$.
When $\kappa=1$, $\rho(0)$ is diagonal in the energy levels, and the system relaxes to the equilibrium Gibbs state.
With $\kappa>1$, coherence emerges in the initial state, and $\pi^{\rm ss}$ becomes a nonequilibrium steady state.
$\kappa$ is varied from $1$ to $4$ and the ratio $S_{\mbb{M}}/\Sigma_\tau$ is plotted as a function of time $\tau$ in Fig.~\ref{fig:num.result.thermal.machine}(b).
As can be seen, the bound is tight and can be saturated for a long time.
Note that the entropy production rate is always positive for $\kappa >1$; therefore, $S_{\mbb{M}}/\Sigma_{\tau}$ goes to zero in the long-time limit.
However, the relation \eqref{eq:main.result.1} is useful because it shows a meaningful bound for initial rapid relaxation processes.

Next, to compare our result with a classical bound, we consider the case that each jump operator characterizes a single jump between two energy levels, $L_{mn}=\sqrt{\eta_{mn}}\Kt{\epsilon_m}\Br{\epsilon_n}~(m\neq n)$.
The transition rates are $\eta_{Ag}=\gamma n_1^{\rm th}(\omega_1)$, $\eta_{gA}=\gamma[n_1^{\rm th}(\omega_1)+1]$, $\eta_{BA}=\gamma n_1^{\rm th}(\omega_3)$, $\eta_{AB}=\gamma[n_1^{\rm th}(\omega_3)+1]$, $\eta_{Bg}=\gamma n_2^{\rm th}(\omega_2)$, and $\eta_{gB}=\gamma[n_2^{\rm th}(\omega_2)+1]$.
In this case, the time evolution of the diagonal terms $p_n(t)\coloneqq\avgObs{\epsilon_n}{\rho(t)}$ follows a classical master equation with time-independent transition rates \cite{Breuer.2002}.
Thereby, Eq.~\eqref{eq:main.result.1} is applied to the diagonal dynamics, which gives $\Sigma_\tau^{\rm cl}\ge D[p_n(0)||p_n(\tau)]$, where $\Sigma_{\tau}^{\rm cl}$ is the entropy production associated with the classical master equation.
In the long-time regime (i.e., when coherence in $\rho(\tau)$ vanishes), it can be proved that $\Sigma_\tau\ge\Sigma_\tau^{\rm cl}$; consequently, $\Sigma_\tau\ge D[p_n(0)||p_n(\tau)]$, which is referred to as the classical bound.
The temperatures and the initial state are the same as in the previous case.
Analogously, $\kappa$ is varied from $1$ to $4$ and the ratios $S_{\mbb{M}}/\Sigma_\tau$ and $D/\Sigma_\tau$ are plotted as functions of time $\tau$ in Fig.~\ref{fig:num.result.thermal.machine}(c).
Two bounds coincide when $\kappa=1$ because there is no coherence in $\rho(t)$ for all $t$.
However, as $\kappa$ increases, the bound $S_{\mbb{M}}$ is tighter than the classical bound.
This is because our bound captures the coherence contribution in the initial state, whereas the classical bound does not.

\emph{Proof of Eq.~\eqref{eq:main.result.1}.}---We first rewrite the dynamics of the density matrix $\rho(t)$ in the interaction picture.
Define $\rho_{I}(t)\coloneqq U_t^\dagger\rho(t)U_t$, the time evolution of $\rho_{I}(t)$ obeys the equation \cite{Supp.PhysRev}
\begin{equation}
\pp_t\rho_{I}(t)=\sum_k\mca{D}_k[\rho_{I}(t)],\label{eq:Int.Lindblad.eq}
\end{equation}
with the initial condition $\rho_{I}(0)=\rho(0)$.
Our approach is based on unraveling the dynamics described by Eq.~\eqref{eq:Int.Lindblad.eq} in terms of quantum trajectories.
In what follows, we demonstrate that the irreversible entropy production $\Sigma_\tau$ can be mathematically linked to the level of individual trajectories.

The framework of quantum trajectories \cite{Horowitz.2012.PRE,Horowitz.2013.NJP} was originally developed in the field of quantum optics as a means to numerically simulate open quantum systems \cite{Breuer.2002}.
Within this approach, the master equation is unraveled into stochastic time evolutions of the pure state of the system $\Kt{\psi(t)}$, conditioned on measurement outcomes obtained from continuous monitoring of the environment.
Each individual trajectory of a stochastic realization can be described by a smooth evolution with discontinuous changes caused by quantum jumps in the state at random times.
A quantum jump is associated with the detection of an event in the environment (e.g., emission or absorption of photons).
The time evolution of the pure state can be described by the stochastic Schr{\"o}dinger equation \cite{Breuer.2002}:
\begin{equation}\label{eq:stoch.Schrodinger.eq}
\begin{aligned}[b]
d\Kt{\psi(t)}&=\mca{S}[\Kt{\psi(t)}]dt\\
&\hspace{0.35cm}+\sum_k\bra{\frac{L_k\Kt{\psi(t)}}{\sqrt{\avg{L_k^\dagger L_k}_{\psi(t)}}}-\Kt{\psi(t)}}dN_{k}(t),
\end{aligned}
\end{equation}
where $\mca{S}(\Kt{\psi})\coloneqq(1/2)\sum_k(\avg{L_k^\dagger L_k}_{\psi}-L_k^\dagger L_k)\Kt{\psi}$ and $\avg{A}_\psi\coloneqq\avgObs{\psi}{A}$.
The stochastic increment $dN_k(t)$ is either $0$ or $1$ (when a jump of type $k$ is detected), and its ensemble average at time $t$ is $\bms{E}[dN_{k}(t)]=\avg{L_k^\dagger L_k}_{\psi(t)}dt$.
Under appropriate initial conditions, the average of $\Mat{\psi(t)}$ over all possible trajectories reduces exactly to the density matrix in the interaction picture, $\bms{E}[\Mat{\psi(t)}]=\rho_{I}(t)$.

Now we define the stochastic entropy production on a single trajectory.
We employ a two-point measurement scheme on the system, where projective measurements are performed at the beginning and at the end of any single trajectory \cite{Manzano.2019.PRL}.
Let $\rho_{I}(0)=\sum_np_n\Mat{n}$ and $\rho_{I}(\tau)=\sum_mp_m\Mat{m}$ be the spectral decompositions of $\rho_{I}(0)$ and $\rho_{I}(\tau)$, the forward process is operated as follows.
The state $\Kt{\psi(0)}$ is sampled from the ensemble $\{\Kt{n}\}$ with probabilities $\{p_n\}$.
The selected state is confirmed by the first measurement in the $\{\Kt{n}\}$ eigenbasis at time $t=0$.
The pure state then evolves in time according to Eq.~\eqref{eq:stoch.Schrodinger.eq}, and the second measurement in the $\{\Kt{m}\}$ eigenbasis is executed at time $t=\tau$.
This procedure results in a stochastic trajectory $\Gamma=\{n,(k_1,t_1),\dots,(k_J,t_J),m\}$, where $n$ and $m$ are measurement outcomes of the first and second measurements, respectively, and $(k_j,t_j)$ denotes a jump of type $k_j$ that occurs at time $t_j$ ($j=1,\dots,J$ and $0\le t_1\le\dots\le t_J\le\tau$).
To define the time-reversed (backward) process, the antiunitary time-reversal operator $\Theta$ is introduced, which satisfies $\Theta i=-i\Theta$ and $\Theta\Theta^\dagger=\Theta^\dagger\Theta=\mbb{I}$.
This operator changes the sign of odd variables under time reversal, such as angular momentum or magnetic fields \cite{Campisi.2011.RMP}.
In the backward process, the initial state $\Kt{\tilde{\psi}(0)}$ is sampled from the ensemble $\{\Kt{\tilde{m}}=\Theta\Kt{m}\}$ with probabilities $\{p_m\}$ and is verified with the projective measurement in the $\{\Kt{\tilde{m}}\}$ eigenbasis.
The pure state $\Kt{\tilde{\psi}(t)}$ analogously obeys Eq.~\eqref{eq:stoch.Schrodinger.eq}, in which the jump operators are replaced by the time-reversed counterparts \cite{Manzano.2018.PRX}
\begin{equation}
\tilde{L}_k=e^{-\Delta s_{\rm env}^k/2}\Theta L_k^\dagger\Theta^\dagger=\Theta L_{k'}\Theta^\dagger.\label{eq:back.jump.operators}
\end{equation}
At time $t=\tau$, the second projective measurement in the $\{\Kt{\tilde{n}}=\Theta\Kt{n}\}$ eigenbasis is performed.
Let $\tilde{\Gamma}=\{m,(k_J,\tau-t_J),\dots,(k_1,\tau-t_1),n\}$ be the time-reversed trajectory that corresponds to $\Gamma$, and the stochastic entropy production associated with the trajectory $\Gamma$ is defined as
\begin{equation}
\Delta s_{\rm tot}[\Gamma]\coloneqq\ln\frac{P(\Gamma)}{\tilde{P}(\tilde{\Gamma})}=\ln\frac{p_n}{p_m}+\sum_{j=1}^J\Delta s_{\rm env}^{k_j},\label{eq:stoch.ent.prod}
\end{equation}
where $P(\Gamma)$ and $\tilde{P}(\tilde{\Gamma})$ are the probabilities of observing the trajectories $\Gamma$ and $\tilde{\Gamma}$, respectively \cite{fnt3}.
Notably, we can prove that the average of $\Delta s_{\rm tot}$ is exactly the irreversible entropy production of the original dynamics [Eq.~\eqref{eq:org.Lindblad.eq}] \cite{Supp.PhysRev},
\begin{equation}
\Sigma_\tau=\avg{\Delta s_{\rm tot}[\Gamma]}=D[P(\Gamma)||\tilde{P}(\tilde{\Gamma})].\label{eq:ent.prod.rel.form}
\end{equation}
Note that the classical relative entropy monotonically decreases under information processing.
Applying the coarse-graining operation $\Lambda:~\Gamma\mapsto n$, which leaves only the first measurement outcome, to $D[P(\Gamma)||\tilde{P}(\tilde{\Gamma})]$, we obtain a lower bound on the entropy production from Eq.~\eqref{eq:ent.prod.rel.form} as
\begin{equation}
\Sigma_\tau\ge D(p_n||\tilde{p}_n),\label{eq:ent.prod.bound.1}
\end{equation}
where $\tilde{p}_n$ is the probability to observe the measurement outcome $\Kt{\tilde{n}}$ at the end of the backward process \cite{fnt1}.
It is crucial to note from Eq.~\eqref{eq:back.jump.operators} that the (inverted) jump operators in the backward process are identical with those in the forward process.
As a consequence, the density matrix in the backward process immediately before the second projective measurement is performed can be explicitly expressed as $\tilde{\rho}_{I}(\tau)=\Theta\rho_{I}(2\tau)\Theta^\dagger$.
The probability distribution $\tilde{p}_n$ can be calculated as $\tilde{p}_n=\davgObs{\tilde{n}}{\tilde{\rho}_{I}(\tau)}{\tilde{n}}=\avgObs{n}{U_{2\tau}^\dagger\rho(2\tau)U_{2\tau}}$.
The inequality \eqref{eq:ent.prod.bound.1} holds for an arbitrary time $\tau>0$ and the entropy production increases in time (i.e., $\Sigma_\tau\ge\Sigma_{\tau/2}$); therefore, Eq.~\eqref{eq:main.result.1} is immediately obtained.

\emph{Conclusion.}---In this Letter, we derived the fundamental bound on irreversibility for thermal relaxation processes of Markovian open quantum systems.
The bound refines the second law of thermodynamics and can be evaluated without knowing details of the underlying dynamics; therefore, it is applicable to the estimation of irreversible entropy production.
Since thermal relaxation is the basis for heat engines, the result of this study is expected to lay the foundations to obtain useful thermodynamic bounds on relevant physical quantities such as power and efficiency \cite{Brandner.2020.PRL}.

\begin{acknowledgments}
The authors acknowledge Keiji Saito, Francesco Buscemi, and Valerio Scarani for fruitful discussions.
The authors are greatly grateful to Keiji Saito for the careful reading of the manuscript and for invaluable comments.
This work was supported by Ministry of Education, Culture, Sports, Science and Technology (MEXT) KAKENHI Grant No. JP19K12153.
\end{acknowledgments}


\begin{thebibliography}{88}
\bibitem [{fnt({\natexlab{b}})}]{fnt1}%
  \BibitemOpen
  \href@noop {} {}Equation \eqref{eq:ent.prod.bound.1} can also be derived
  using an integral fluctuation theorem \cite{Supp.PhysRev}.\BibitemShut
  {Stop}%
\end{thebibliography}
\end{document}


\title{Supplemental Material for \\ ``Lower Bound on Irreversibility in Thermal Relaxation of Open Quantum Systems''}

\author{Tan Van Vu}
\email{tan@biom.t.u-tokyo.ac.jp}
\affiliation{Department of Information and Communication Engineering, Graduate
	School of Information Science and Technology, The University of Tokyo,
	Tokyo 113-8656, Japan}

\author{Yoshihiko Hasegawa}
\email{hasegawa@biom.t.u-tokyo.ac.jp}
\affiliation{Department of Information and Communication Engineering, Graduate
	School of Information Science and Technology, The University of Tokyo,
	Tokyo 113-8656, Japan}

\pacs{}
\maketitle

This supplemental material describes the details of calculations introduced in the main text. 
The equations and figure numbers are prefixed with S [e.g., Eq.~(S1) or Fig.~S1]. 
Numbers without this prefix [e.g., Eq.~(1) or Fig.~1] refer to items in the main text.



\section{Entropy production in time-driven systems cannot be bounded by a universal metric}
Here we prove that irreversible entropy production cannot be bounded from below by a system-independent distance between the initial and final states.
It is enough to prove for the classical case; therefore, we consider time-dependent Markov jump processes described by the master equation
\begin{equation}
\pp_t p_n(t)=\sum_{m(\neq n)}[R_{nm}(t)p_m(t)-R_{mn}(t)p_n(t)],
\end{equation}
where $R_{mn}(t)\ge 0$ denotes the time-dependent transition rate from state $n$ to state $m$.
Let $\mbm{p}(t)\coloneqq[p_1(t),\dots,p_N(t)]^\top$ be the time-dependent probability distribution.
With proof by contradiction, we assume that there exists a metric $\ell$ that is independent of the system parameters and satisfies
\begin{equation}
\Sigma_\tau\ge\ell[\mbm{p}(0),\mbm{p}(\tau)]
\end{equation}
for all Markov jump processes.
The irreversible entropy production can be written as
\begin{equation}
\Sigma_\tau=\frac{1}{2}\int_0^\tau\sum_{m\neq n}[R_{mn}(t)p_n(t)-R_{nm}(t)p_m(t)]\ln\frac{R_{mn}(t)p_n(t)}{R_{nm}(t)p_m(t)}dt.
\end{equation}
We consider the nontrivial case: $\ell[\mbm{p}(0),\mbm{p}(\tau)]>0$.
Let $\delta>0$ be an arbitrary number that satisfies $\ell[\mbm{p}(0),\mbm{p}(\tau)]>\delta$.
We define auxiliary dynamics with transition rates:
\begin{equation}
\tilde{R}_{mn}(t)\coloneqq R_{mn}(t)+\frac{\alpha}{p_n(t)},
\end{equation}
where $\alpha>0$ is a positive number that will be determined later.
The initial distribution of the auxiliary dynamics is set to $\tilde{\mbm{p}}(0)=\mbm{p}(0)$.
It can then be proved that the probability distribution of the auxiliary dynamics is the same as the original for all times, i.e., $\tilde{\mbm{p}}(t)=\mbm{p}(t)\,\forall t\ge 0$.
Specifically, we show that if $\tilde{\mbm{p}}(t)=\mbm{p}(t)$, then $\tilde{\mbm{p}}(t+\Delta t)=\mbm{p}(t+\Delta t)$ for arbitrarily small $\Delta t>0$.
From the master equation, we have
\begin{subequations}
\begin{align}
\tilde{p}_n(t+\Delta t)-\tilde{p}_n(t)&=\Delta t\sum_{m(\neq n)}\bras{\tilde{R}_{nm}(t)\tilde{p}_m(t)-\tilde{R}_{mn}(t)\tilde{p}_n(t)}\\
&=\Delta t\sum_{m(\neq n)}\bras{\tilde{R}_{nm}(t)p_m(t)-\tilde{R}_{mn}(t)p_n(t)}\\
&=\Delta t\sum_{m(\neq n)}\bras{R_{nm}(t)p_m(t)-R_{mn}(t)p_n(t)}\\
&=p_n(t+\Delta t)-p_n(t),
\end{align}
\end{subequations}
which implies that $\tilde{p}_n(t+\Delta t)=p_n(t+\Delta t)$ for all $n$.
The original and auxiliary dynamics thus have the same distributions for all times $0\le t\le\tau$.
Therefore, the entropy production $\tilde{\Sigma}_\tau$ in the auxiliary dynamics is also bounded from below by the distance $\ell$, as
\begin{equation}
\tilde{\Sigma}_\tau\ge\ell[\mbm{p}(0),\mbm{p}(\tau)]>\delta.\label{eq:eps.bound.ent}
\end{equation}
However,
\begin{subequations}
\begin{align}
\tilde{\Sigma}_\tau &=\frac{1}{2}\int_0^\tau\sum_{m\neq n}[\tilde{R}_{mn}(t)\tilde{p}_n(t)-\tilde{R}_{nm}(t)\tilde{p}_m(t)]\ln\frac{\tilde{R}_{mn}(t)\tilde{p}_n(t)}{\tilde{R}_{nm}(t)\tilde{p}_m(t)}dt\\
&=\frac{1}{2}\int_0^\tau\sum_{m\neq n}[R_{mn}(t)p_n(t)-R_{nm}(t)p_m(t)]\ln\frac{R_{mn}(t)p_n(t)+\alpha}{R_{nm}(t)p_m(t)+\alpha}dt\\
&\le\frac{1}{2\alpha}\int_0^\tau\sum_{m\neq n}[R_{mn}(t)p_n(t)-R_{nm}(t)p_m(t)]^2dt.\label{eq:ent.upper}
\end{align}
\end{subequations}
To obtain Eq.~\eqref{eq:ent.upper}, we have applied the inequality
\begin{equation}
(x-y)\ln\frac{x+z}{y+z}\le\frac{(x-y)^2}{z}
\end{equation}
for $x,y,z\ge 0$.
Now, we select a sufficiently large value of $\alpha$, such that
\begin{equation}
\frac{1}{2\alpha}\int_0^\tau\sum_{m\neq n}[R_{mn}(t)p_n(t)-R_{nm}(t)p_m(t)]^2dt<\delta\Leftrightarrow\alpha>\frac{1}{2\delta}\int_0^\tau\sum_{m\neq n}[R_{mn}(t)p_n(t)-R_{nm}(t)p_m(t)]^2dt.
\end{equation}
The following inequality is then obtained from Eq.~\eqref{eq:ent.upper}:
\begin{equation}
\tilde{\Sigma}_\tau<\delta.
\end{equation}
This is inconsistent with Eq.~\eqref{eq:eps.bound.ent}, which completes our proof.

\section{Derivation of Eq.~(\IntLindbladEq) in the main text}
For any operators $A$ and $B$, one can prove that
\begin{equation}
e^{-\lambda A}Be^{\lambda A}=\sum_{n=0}^{\infty}\frac{(-\lambda)^n}{n!}[A,B]_n,
\end{equation}
where the nested commutator is recursively defined as $[A,B]_n=[A,[A,B]_{n-1}]$ and $[A,B]_0=B$.
From the relation $[L_k,H]=\omega_kL_k$, we can readily obtain
\begin{equation}
[H,L_k]_n=(-\omega_k)^nL_k.
\end{equation}
Consequently,
\begin{equation}
e^{-\lambda H}L_ke^{\lambda H}=\sum_{n=0}^{\infty}\frac{(-\lambda)^n(-\omega_k)^n}{n!}L_k=e^{\lambda\omega_k}L_k\Rightarrow e^{-\lambda H}L_k=e^{\lambda\omega_k}L_ke^{-\lambda H}.
\end{equation}
The time derivative of $\rho_{I}(t)=e^{iHt}\rho(t)e^{-iHt}$ is then taken to derive Eq.~(\IntLindbladEq) as
\begin{subequations}
\begin{align}
\pp_t{\rho}_{I}(t)&=e^{iHt}i[H,\rho(t)]e^{-iHt}+e^{iHt}\pp_t{\rho}(t)e^{-iHt}\\
&=\sum_k e^{iHt}\bras{L_k\rho(t)L_k^\dagger-\frac{1}{2}\{L_k^\dagger L_k,\rho(t)\}}e^{-iHt}\\
&=\sum_k \bras{L_ke^{iHt}\rho(t)e^{-iHt}L_k^\dagger-\frac{1}{2}\{L_k^\dagger L_k,e^{iHt}\rho(t)e^{-iHt}\}}\\
&=\sum_k \bras{L_k\rho_{I}(t)L_k^\dagger-\frac{1}{2}\{L_k^\dagger L_k,\rho_{I}(t)\}}\\
&=\sum_k\mca{D}_k[\rho_{I}(t)].
\end{align}
\end{subequations}
Here, we have used the following relations:
\begin{align}
e^{iHt}L_k&=e^{-i\omega_kt}L_ke^{iHt},\\
L_k^\dagger e^{-iHt}&=e^{i\omega_kt}e^{-iHt}L_k^\dagger,\\
[L_k^\dagger L_k,H]&=0.
\end{align}

\section{Quantum trajectories and entropy production}
Here we show that the irreversible entropy production $\Sigma_\tau$ can be mapped to the stochastic entropy production on the level of individual trajectories.
First, Eq.~(\IntLindbladEq) can be rewritten as
\begin{equation}
\pp_t\rho_{I}(t)=-i[H_{\rm eff}\rho_{I}(t)-\rho_{I}(t)H_{\rm eff}^\dagger]+\sum_kL_k\rho_{I}(t)L_k^\dagger,
\end{equation}
where $H_{\rm eff}\coloneqq -(i/2)\sum_kL_k^\dagger L_k$ is the skew-Hermitian effective Hamiltonian (i.e., $H_{\rm eff}^\dagger=-H_{\rm eff}$).
In the forward process, the evolution of the pure state between jumps is described by the deterministic equation
\begin{equation}
\frac{d}{dt}\Kt{\psi(t)}=\mca{S}[\Kt{\psi(t)}].\label{eq:det.evol.eq}
\end{equation}
The formal solution of Eq.~\eqref{eq:det.evol.eq} gives the state at time $t~(>s)$ as
\begin{equation}
\Kt{\psi(s)}\mapsto\Kt{\psi(t)}=\frac{U_{\rm eff}(t,s)\Kt{\psi(s)}}{\sqrt{\avg{U_{\rm eff}(t,s)^\dagger U_{\rm eff}(t,s)}_{\psi(s)}}},
\end{equation}
where the effective time-evolution operator $U_{\rm eff}(t,s)=e^{-iH_{\rm eff}(t-s)}$ is the solution of the differential equation
\begin{equation}
\pp_tU_{\rm eff}(t,s)=-iH_{\rm eff}U_{\rm eff}(t,s),
\end{equation}
with the initial condition $U_{\rm eff}(s,s)=\mbb{I}$.
The smooth evolution of the pure state is interrupted by sudden jumps, which alter the state as
\begin{equation}
\Kt{\psi(t)}\mapsto\frac{L_k\Kt{\psi(t)}}{\sqrt{\avg{L_k^\dagger L_k}_{\psi(t)}}}.
\end{equation}
This discontinuous change is induced by the jump operator $\mca{J}_k\coloneqq L_k\sqrt{dt}$.
Given the stochastic trajectory $\Gamma=\{n,(k_1,t_1),\dots,(k_J,t_J),m\}$, the probability to observe $\Gamma$ is encoded into the unnormalized state
\begin{subequations}
\begin{align}
\Kt{\psi_\tau(\Gamma)}&=\Mat{m}U_{\rm eff}(\tau,t_J)\mca{J}_{k_J}\dots\mca{J}_{k_1}U_{\rm eff}(t_1,0)\Kt{n}\\
&=\Mat{m}\mca{P}(\Gamma)\Kt{n},
\end{align}
\end{subequations}
where $\mca{P}(\Gamma)\coloneqq U_{\rm eff}(\tau,t_J)\mca{J}_{k_J}\dots\mca{J}_{k_1}U_{\rm eff}(t_1,0)$ is the forward propagator.
In other words, given that the initial state is $\Kt{n}$, the probability of observing $\Gamma$ is the norm of $\Kt{\psi_\tau(\Gamma)}$,
\begin{equation}
P(\Gamma|n)=\Prod{\psi_\tau(\Gamma)}=|\davgObs{m}{\mca{P}(\Gamma)}{n}|^2.
\end{equation}
With the backward process defined as in the main text, the effective Hamiltonian is $\tilde{H}_{\rm eff}=-(i/2)\sum_k\tilde{L}_k^\dagger\tilde{L}_k=\Theta (i/2)\sum_kL_k^\dagger L_k\Theta^\dagger=-\Theta H_{\rm eff}\Theta^\dagger$ and the corresponding effective time-evolution operator is $\tilde{U}_{\rm eff}(t,s)=e^{-i\tilde{H}_{\rm eff}(t-s)}$.
In addition, the backward jump operators are
\begin{equation}
\tilde{\mca{J}}_k\coloneqq\tilde{L}_k\sqrt{dt}=\Theta\mca{J}_{k'}\Theta^\dagger.
\end{equation}
Analogously, the probability of observing the time-reversed trajectory $\tilde{\Gamma}=\{m,(k_J,\tau-t_J),\dots,(k_1,\tau-t_1),n\}$ in the backward process is encoded into the unnormalized state
\begin{subequations}
\begin{align}
\Kt{\tilde{\psi}_\tau(\tilde{\Gamma})}&=\Mat{\tilde{n}}\tilde{U}_{\rm eff}(\tau,\tau-t_1)\tilde{\mca{J}}_{k_1}\dots\tilde{\mca{J}}_{k_J}\tilde{U}_{\rm eff}(\tau-t_J,0)\Kt{\tilde{m}}\\
&=\Mat{\tilde{n}}\tilde{\mca{P}}(\tilde{\Gamma})\Kt{\tilde{m}},
\end{align}
\end{subequations}
where $\tilde{\mca{P}}(\tilde{\Gamma})\coloneqq \tilde{U}_{\rm eff}(\tau,\tau-t_1)\tilde{\mca{J}}_{k_1}\dots\tilde{\mca{J}}_{k_J}\tilde{U}_{\rm eff}(\tau-t_J,0)$ is the backward propagator.
The probability to observe $\tilde{\Gamma}$ given that the initial state is $\Kt{\tilde{m}}$ is the norm of $\Kt{\tilde{\psi}_\tau(\tilde{\Gamma})}$,
\begin{equation}
\tilde{P}(\tilde{\Gamma}|\tilde{m})=\Prod{\tilde{\psi}_\tau(\tilde{\Gamma})}=|\davgObs{\tilde{n}}{\tilde{\mca{P}}(\tilde{\Gamma})}{\tilde{m}}|^2.
\end{equation}
It can be confirmed that
\begin{equation}
\Theta^\dagger\tilde{U}_{\rm eff}^\dagger(\tau-s,\tau-t)\Theta=\Theta^\dagger e^{i\tilde{H}_{\rm eff}^\dagger(t-s)}\Theta=\Theta^\dagger e^{i\Theta H_{\rm eff}\Theta^\dagger(t-s)}\Theta= e^{-iH_{\rm eff}(t-s)}=U_{\rm eff}(t,s).
\end{equation}
Consequently, the propagators in the forward and backward processes can be related as
\begin{equation}
\mca{P}(\Gamma)=\Theta^\dagger\tilde{\mca{P}}^\dagger(\tilde{\Gamma})\Theta e^{\sum_{j=1}^J\Delta s_{\rm env}^{k_j}/2}\Rightarrow\davgObs{m}{\mca{P}(\Gamma)}{n}=\davgObs{\tilde{n}}{\tilde{\mca{P}}(\tilde{\Gamma})}{\tilde{m}}^*e^{\sum_{j=1}^J\Delta s_{\rm env}^{k_j}/2}.
\end{equation}
From the relations $P(\Gamma)=P(\Gamma|n)p_n$ and $\tilde{P}(\tilde{\Gamma})=\tilde{P}(\tilde{\Gamma}|\tilde{m})p_m$, we immediately have
\begin{equation}
\Delta s_{\rm tot}=\ln\frac{P(\Gamma)}{\tilde{P}(\tilde{\Gamma})}=\ln\frac{p_n}{p_m}+\sum_{j=1}^J\Delta s_{\rm env}^{k_j}.
\end{equation}

Finally, we show that the ensemble average of $\Delta s_{\rm tot}$ is equal to the irreversible entropy production $\Sigma_\tau$.
Noting that $\rho_{I}(t)=e^{iHt}\rho(t)e^{-iHt}$ and the von Neumann entropy is invariant under unitary transforms, we have
\begin{subequations}
\begin{align}
\avg{\Delta s_{\rm tot}}&=\avg{-\ln p_m+\ln p_n+\sum_{j=1}^J\Delta s_{\rm env}^{k_j}}\\
&=-\sum_mp_m\ln p_m+\sum_np_n\ln p_n+\int_0^\tau\sum_k\Tr{L_k^\dagger L_k\rho_{I}(t)}\Delta s_{\rm env}^{k}dt\\
&=-\Tr{\rho_{I}(\tau)\ln\rho_{I}(\tau)}+\Tr{\rho_{I}(0)\ln\rho_{I}(0)}+\Delta S_{\rm env}\\
&=-\Tr{\rho(\tau)\ln\rho(\tau)}+\Tr{\rho(0)\ln\rho(0)}+\Delta S_{\rm env}\\
&=\Delta S_{\rm sys}+\Delta S_{\rm env}\\
&=\Sigma_\tau.
\end{align}
\end{subequations}

\section{Derivation of Eq.~(\KLbound) based on an integral fluctuation theorem}
Here we derive Eq.~(\KLbound) using a fluctuation theorem.
To this end, we prove the following equality:
\begin{equation}
\avgl{\exp\bras{-\Delta s_{\rm tot}+\ln\frac{p_n}{\tilde{p}_n}}}=1,\label{eq:FTlike}
\end{equation}
where the average is taken with respect to the distribution $P(\Gamma)$.
Inserting $\Delta s_{\rm tot}=\ln[P(\Gamma)/\tilde{P}(\tilde{\Gamma})]$, we have
\begin{subequations}
\begin{align}
\avgl{\exp\bras{-\Delta s_{\rm tot}+\ln\frac{p_n}{\tilde{p}_n}}}&=\avgl{\exp\bras{-\ln\frac{P(\Gamma|n)p_n}{\tilde{P}(\tilde{\Gamma}|\tilde{m})p_m}+\ln\frac{p_n}{\tilde{p}_n}}}\\
&=\sum_\Gamma P(\Gamma|n)p_n\frac{\tilde{P}(\tilde{\Gamma}|\tilde{m})p_m}{P(\Gamma|n)\tilde{p}_n}\\
&=\sum_n \frac{p_n}{\tilde{p}_n}\sum_{\Gamma|n}\tilde{P}(\tilde{\Gamma}|\tilde{m})p_m\\
&=\sum_n \frac{p_n}{\tilde{p}_n}\tilde{p}_n \label{eq:FT.deriv.tmp1}\\
&=\sum_n p_n\\
&=1.
\end{align}
\end{subequations}
To obtain Eq.~\eqref{eq:FT.deriv.tmp1}, we have used $\sum_{\Gamma|n}\tilde{P}(\tilde{\Gamma}|\tilde{m})p_m=\tilde{p}_n$.
Jensen's inequality is then applied to Eq.~\eqref{eq:FTlike} to obtain
\begin{equation}
\Sigma_\tau=\avg{\Delta s_{\rm tot}}\ge\avgl{\ln\frac{p_n}{\tilde{p}_n}}=D(p_n||\tilde{p}_n).
\end{equation}

\section{Derivation of Eq.~(\EigenBound) in the main text}
In the single-reservoir case, irreversible entropy production can be explicitly written as $\Sigma_\tau=S[\rho(0)||\pi]-S[\rho(\tau)||\pi]\ge 0$.
It should be noted that the straight-forward extension of the classical bound, $\Sigma_\tau\ge S[\rho(0)||\rho(\tau)]$, does not generally hold.
Its violation can be intuitively confirmed in the limit of vanishing coupling to the reservoir \cite{Vu.2021.PRL} and has also been experimentally verified for a single-atom system \cite{Zhang.2020.PRR}.
The interaction-picture Lindblad master equation has no unitary part; therefore, it can be proved that \cite{Alhambra.2017.PRA}
\begin{equation}
S[\rho_I(0)||\pi]-S[\rho_I(\tau/2)||\pi]\ge S[\rho_I(0)||\rho_{I}(\tau)].\label{eq:equ.relax.rel.bound.1}
\end{equation}
Due to the invariance of the relative entropy under unitary transforms, the term in the left-hand side of Eq.~\eqref{eq:equ.relax.rel.bound.1} equals $\Sigma_{\tau/2}$.
Since $\Sigma_\tau\ge\Sigma_{\tau/2}$ and $S[\rho_I(0)||\rho_{I}(\tau)]=S[U_\tau\rho(0)U_\tau^\dagger||\rho(\tau)]$, a bound on the entropy production can be obtained, $\Sigma_\tau\ge S[U_\tau\rho(0)U_\tau^\dagger||\rho(\tau)]$.
The dependence of this lower bound on the Hamiltonian can be eliminated by taking the minimum over all unitary operators, which yields the following bound:
\begin{equation}
\Sigma_\tau\ge\min_{V^\dagger V=\mbb{I}}S[V\rho(0)V^\dagger||\rho(\tau)].\label{eq:H.free.bound}
\end{equation}
We prove below that the variational term in the right-hand side of Eq.~\eqref{eq:H.free.bound} is exactly $S_{\mbb{E}}[\rho(0)||\rho(\tau)]$, from which Eq.~(\EigenBound) is obtained.

In what follows, we prove that $S_{\mbb{E}}(\rho||\sigma)=\min_{V^\dagger V=\mbb{I}}S(V\rho V^\dagger||\sigma)$.
Note that $S_{\mbb{E}}(\rho||\sigma)=D(a_n||b_n)$, where $\{a_n\}_n$ and $\{b_n\}_n$ are the increasing eigenvalues of $\rho$ and $\sigma$.
To this end, we will show that $\min_{V^\dagger V=\mbb{I}}S(V\rho V^\dagger||\sigma)\le D(a_n||b_n)$ and $\min_{V^\dagger V=\mbb{I}}S(V\rho V^\dagger||\sigma)\ge D(a_n||b_n)$.
First, we prove the former.
Let $\rho=\sum_n a_n\Mat{a_n}$ and $\sigma=\sum_n b_n\Mat{b_n}$ be the spectral decompositions of $\rho$ and $\sigma$, respectively.
Note that the matrix $V=\sum_{n}\dMat{b_n}{a_n}$ is a unitary matrix, by which we readily obtain
\begin{equation}
\min_{V^\dagger V=\mbb{I}}S(V\rho V^\dagger||\sigma)\le S(V\rho V^\dagger||\sigma)=S(\sum_n a_n\Mat{b_n}~||~\sum_n b_n\Mat{b_n})=D(a_n||b_n).
\end{equation}
Next, we prove the latter.
For an arbitrary unitary matrix $V$, we have
\begin{subequations}
\begin{align}
S(V\rho V^\dagger||\sigma)&=\Tr{\rho\ln\rho}-\Tr{V\rho V^\dagger\ln\sigma}\\
&=\sum_n a_n\ln a_n-\sum_{n,m}a_n\ln b_m|\davgObs{b_m}{V}{a_n}|^2\\
&=\sum_n a_n\ln a_n-\sum_{n,m}c_{nm}a_n\ln b_m,\label{eq:rel.decomp}
\end{align}
\end{subequations}
where $c_{nm}\coloneqq |\davgObs{b_m}{V}{a_n}|^2\ge 0$.
Note that $\sum_m c_{nm}=\sum_n c_{nm}=1$.

Before proceeding further, we prove the following Lemma.
\begin{lemma}
Let $0\le a_1\le\dots\le a_N$ and $0\le b_1\le\dots\le b_N$ be two arrays of nonnegative numbers and $C=[c_{nm}]\in\mbb{R}^{N\times N}$ be a doubly stochastic matrix (i.e., $c_{nm}\ge 0$ and $\sum_mc_{nm}=\sum_nc_{nm}=1$).
Then, 
\begin{equation}
\sum_{n,m}c_{nm}a_n\ln b_m\le\sum_na_n\ln b_n.
\end{equation}
\end{lemma}
\begin{proof}
We prove by induction. The $N=1$ case is trivial.
Assuming that the result holds for $N=k-1$ with $k\ge 2$, we show it also holds for $N=k$.
For convenience, we define $\mca{F}(C)\coloneqq\sum_{n=1}^k\sum_{m=1}^kc_{nm}a_n\ln b_m$ as the function associated with the matrix $C$.
We need only prove $\mca{F}(C)\le\sum_{n=1}^ka_n\ln b_n$.
If $c_{kk}<1$, then there exist two indexes $1\le i,j\le k-1$ such that $c_{ik}>0$ and $c_{kj}>0$.
We set $\epsilon=\min(c_{ik},c_{kj})$ and note that
\begin{equation}
a_{i}\ln b_{j} + a_k\ln b_k\ge a_{i}\ln b_k + a_k\ln b_{j}~[\because (a_k-a_{i})(\ln b_k-\ln b_{j})\ge 0].
\end{equation}
We define a new matrix $C'=[c_{nm}']\in\mbb{R}^{k\times k}$ as 
\begin{equation}
c_{kk}'=c_{kk}+\epsilon,~c_{ij}'=c_{ij}+\epsilon,~c_{ik}'=c_{ik}-\epsilon,~c_{kj}'=c_{kj}-\epsilon,
\end{equation}
and $c_{nm}'=c_{nm}$ otherwise.
All elements of $C'$ are nonnegative and satisfy $\sum_mc_{nm}'=\sum_nc_{nm}'=1$.
This implies that this procedure generates a new doubly stochastic matrix that yields a larger value of the function $\mca{F}$, i.e., $\sum_{n,m}c_{nm}a_n\ln b_m=\mca{F}(C)\le\mca{F}(C')=\sum_{n,m}c_{nm}'a_n\ln b_m$.
Repeating this procedure a finite number of times, we eventually obtain a doubly stochastic matrix $C_{\rm fin}$ with $c_{kk}'=1$.
Then, $c_{kn}'=c_{nk}'=0$ for all $n=1,\dots,k-1$ and the submatrix $C_{\rm sub}=[c_{nm}']\in\mbb{R}^{(k-1)\times(k-1)}$, which is obtained by eliminating the $k$th row and the $k$th column of $C_{\rm fin}$, is also a doubly stochastic matrix; therefore, $\mca{F}(C_{\rm sub})\le\sum_{n=1}^{k-1}a_n\ln b_n$.
Consequently,
\begin{equation}
\mca{F}(C)\le\mca{F}(C_{\rm fin})=\mca{F}(C_{\rm sub})+a_k\ln b_k\le\sum_{n=1}^ka_n\ln b_n,
\end{equation}
which completes our proof.
\end{proof}

According to Lemma 1, we have $\sum_{n,m}c_{nm}a_n\ln b_m\le\sum_n a_n\ln b_n$.
Consequently, Eq.~\eqref{eq:rel.decomp} implies
\begin{equation}
S(V\rho V^\dagger||\sigma)\ge \sum_n a_n\ln a_n-\sum_n a_n\ln b_n = D(a_n||b_n).\label{eq:proof.tmp1}
\end{equation}
Equation~\eqref{eq:proof.tmp1} holds for an arbitrary unitary matrix $V$; therefore, we obtain $\min_{V^\dagger V=\mbb{I}}S(V\rho V^\dagger||\sigma)\ge D(a_n||b_n)$.

\section{Another derivation of Eq.~(\EigenBound) for a two-level atom system}
We consider the thermal relaxation process of a two-level atom, which is weakly coupled to a thermal reservoir.
The Hamiltonian of the system is $H=\omega\sigma_z/2$.
The time evolution of the density matrix obeys the Lindblad equation:
\begin{equation}
\pp_t{\rho}=-i[H,\rho]+\gamma\bar{n}(\omega)(\sigma_+\rho\sigma_--\frac{1}{2}\{\sigma_-\sigma_+,\rho\})+\gamma(\bar{n}(\omega)+1)(\sigma_-\rho\sigma_+-\frac{1}{2}\{\sigma_+\sigma_-,\rho\}),
\end{equation}
where $\sigma_{\pm}=(\sigma_x\pm i\sigma_y)/2$, $\gamma$ is a positive damping rate, and $\bar{n}(\omega)=(e^{\beta\omega}-1)^{-1}$ is the Planck distribution.
The density operator $\rho(t)$ during the relaxation process is analytically solvable and the irreversible entropy production can be explicitly evaluated as $\Sigma_\tau=S[\rho(0)||\pi]-S[\rho(\tau)||\pi]$, where $\tau$ denotes the process time.

The density matrix can be represented using the Bloch sphere
\begin{equation}
\rho(t)=\frac{1}{2}\bras{\mbb{I}+\mbm{r}(t)\cdot\mbm{\sigma}},
\end{equation}
where $\mbm{r}(t)\coloneqq [r_x(t),r_y(t),r_z(t)]^\top$ is the Bloch vector and $\mbm{\sigma}\coloneqq [\sigma_x,\sigma_y,\sigma_z]^\top$ denotes the vector of the Pauli matrices.
Note that $r(t)^2\coloneqq r_x(t)^2+r_y(t)^2+r_z(t)^2\le 1$.
The density matrix can be explicitly calculated as
\begin{align}
r_x(t)&=e^{-\overline{\gamma}\tau/2}\bras{r_x(0)\cos(\omega\tau)-r_y(0)\sin(\omega\tau)},\\
r_y(t)&=e^{-\overline{\gamma}\tau/2}\bras{r_x(0)\sin(\omega\tau)+r_y(0)\cos(\omega\tau)},\\
r_z(t)&=r_z(0)e^{-\overline{\gamma}\tau}+\tanh(\beta\omega/2)[e^{-\overline{\gamma}\tau}-1],
\end{align}
where $\overline{\gamma}\coloneqq \gamma\coth(\beta\omega/2)$.
The eigenvalues of $\rho(t)$ are $[1\pm r(t)]/2$; therefore, the von Neumman entropy $S(t)=-\Tr{\rho(t)\ln\rho(t)}$ can be expressed in terms of the magnitude of the Bloch vector as
\begin{equation}
S(t)=-\frac{1-r(t)}{2}\ln\frac{1-r(t)}{2}-\frac{1+r(t)}{2}\ln\frac{1+r(t)}{2}.
\end{equation}
The irreversible entropy production can be written as
\begin{subequations}
\begin{align}
\Sigma_\tau&=S(\tau)-S(0)+\beta\Tr{H[\rho(0)-\rho(\tau)]}\\
&=S(\tau)-S(0)+\beta\omega[r_z(0)-r_z(\tau)]/2.
\end{align}
\end{subequations}
In the following, we prove the inequality [Eq.~(\EigenBound)]:
\begin{equation}
\Sigma_\tau\ge D[p_n(0)||p_n(\tau)],\label{eq:info.bound}
\end{equation}
where $\{p_n(t)\}_{n=1}^N$ are increasing eigenvalues of $\rho(t)$.

Equation \eqref{eq:info.bound} is equivalent to
\begin{equation}
\beta\omega[r_z(0)-r_z(\tau)]\ge[r(\tau)-r(0)]\ln\frac{1+r(\tau)}{1-r(\tau)}.\label{eq:simp.bound}
\end{equation}
For convenience, we set $a\coloneqq e^{-\overline{\gamma}\tau}\in(0,1)$, $r_{\rm eq}\coloneqq \tanh(\beta\omega/2)\in(0,1)$, and $\kappa\coloneqq r_x(0)^2+r_y(0)^2\ge 0$.
Then, $r(0)=\sqrt{\kappa+r_z(0)^2}$ and $r(\tau)=\sqrt{a\kappa+[ar_z(0)+(a-1)r_{\rm eq}]^2}$.
To prove Eq.~\eqref{eq:simp.bound}, we divide into two cases of $r(\tau)\le r(0)$ and $r(\tau)>r(0)$. 

\underline{\underline{(a) $r(\tau)\le r(0)$:}} Setting $f(\kappa)\coloneqq [r(\tau)-r(0)]\ln\brab{\bras{1+r(\tau)}/\bras{1-r(\tau)}}$ as a function of $\kappa$, Eq.~\eqref{eq:simp.bound} can be rewritten as
\begin{equation}
f(\kappa)\le\beta\omega[r_z(0)-r_z(\tau)].\label{eq:simp.tmp1}
\end{equation}
Since $f(\kappa)\le 0$, we need only consider the $r_z(0)<r_z(\tau)$ case.
To prove Eq.~\eqref{eq:simp.tmp1}, we first prove that $f(\kappa)$ is a decreasing function.
Specifically, we show that $df(\kappa)/d\kappa\le 0$.
Taking the derivative of $f(\kappa)$ with respect to $\kappa$, we obtain
\begin{equation}
\frac{d}{d\kappa}f(\kappa)=\frac{ar(0)-r(\tau)}{r(0)}\ln\frac{1+r(\tau)}{1-r(\tau)}+\frac{2a[r(\tau)-r(0)]}{1-r(\tau)^2}.\label{eq:diff1}
\end{equation}
If $ar(0)-r(\tau)\le 0$, then $df(\kappa)/d\kappa\le 0$.
If $ar(0)-r(\tau)>0$, then we apply the inequality
\begin{equation}
\ln\frac{1+r(\tau)}{1-r(\tau)}\le\frac{2r(\tau)}{1-r(\tau)^2}
\end{equation}
to Eq.~\eqref{eq:diff1}, which gives
\begin{subequations}
\begin{align}
\frac{d}{d\kappa}f(\kappa)&\le\frac{2r(\tau)[ar(0)-r(\tau)]}{r(0)[1-r(\tau)^2]}+\frac{2a[r(\tau)-r(0)]}{1-r(\tau)^2}\\
&=\frac{2r(\tau)[ar(0)-r(\tau)]+2ar(0)[r(\tau)-r(0)]}{r(0)[1-r(\tau)^2]}.
\end{align}
\end{subequations}
Since $r(\tau)\le r(0)$ and $ar(0)-r(\tau)\le a[r(0)-r(\tau)]$, we can readily obtain $2r(\tau)[ar(0)-r(\tau)]+2ar(0)[r(\tau)-r(0)]\le 0$, or equivalently, $df(\kappa)/d\kappa\le 0$.

From $r_z(0)<ar_z(0)+(a-1)r_{\rm eq}=r_z(\tau)$, $r_z(0)+r_{\rm eq}<0$ is obtained.
In addition, from $r(0)^2=\kappa+r_z(0)^2\ge a\kappa+[ar_z(0)+(a-1)r_{\rm eq}]^2=r(\tau)^2$, we can derive
\begin{equation}
\kappa\ge\max\brab{0,[r_{\rm eq}+r_z(0)][(1-a)r_{\rm eq}-(1+a)r_z(0)]}=0.
\end{equation}
Since $f(\kappa)$ is a decreasing function, we have $f(\kappa)\le f(0)$, which is equivalent to
\begin{subequations}
\begin{align}
f(\kappa)&\le[|r_z(\tau)|-|r_z(0)|]\ln\frac{1+|r_z(\tau)|}{1-|r_z(\tau)|}\\
&=[r_z(0)-r_z(\tau)]\ln\frac{1+|r_z(\tau)|}{1-|r_z(\tau)|}.
\end{align}
\end{subequations}
Here, we have used $r_z(0)<0$ and $r_z(\tau)=ar_z(0)+(a-1)r_{\rm eq}\le 0$.
Since $|r_z(\tau)|=(1-a)r_{\rm eq}-ar_z(0)\ge r_{\rm eq}$ and $\ln[(1+x)/(1-x)]$ is an increasing function of $x\in[0,1]$, we further obtain
\begin{equation}
f(\kappa)\le [r_z(0)-r_z(\tau)]\ln\frac{1+r_{\rm eq}}{1-r_{\rm eq}}=\beta\omega[r_z(0)-r_z(\tau)],
\end{equation}
which is exactly Eq.~\eqref{eq:simp.tmp1}.

\underline{\underline{(b) $r(\tau)>r(0)$:}} The condition $r(\tau)>r(0)$ is equivalent to
\begin{equation}
0\le\kappa\le(r_z(0)+r_{\rm eq})[(1-a)r_{\rm eq}-(a+1)r_z(0)],\label{eq:kappa.cond}
\end{equation}
from which we can derive $r_z(0)+r_{\rm eq}\ge 0$.
Consequently, $r_z(0)-r_z(\tau)=(1-a)[r_z(0)+r_{\rm eq}]\ge 0$.
Therefore, Eq.~\eqref{eq:simp.bound} can be rewritten as
\begin{equation}
\beta\omega\ge\frac{r(\tau)-r(0)}{r_z(0)-r_z(\tau)}\ln\frac{1+r(\tau)}{1-r(\tau)}\eqqcolon g(a).\label{eq:simp.tmp2}
\end{equation}
To prove this, we first show that $g(a)$ is a decreasing function with respect to $a$.
Specifically, we prove that $g_1(a)\coloneqq \ln\brab{[1+r(\tau)]/[1-r(\tau)]}$ and $g_2(a)\coloneqq [r(\tau)-r(0)]/[r_z(0)-r_z(\tau)]$ are decreasing functions with respect to $a$.
Noting that
\begin{subequations}
\begin{align}
\frac{d}{da}r(\tau)&=\frac{\kappa+2[ar_z(0)+(a-1)r_{\rm eq}][r_z(0)+r_{\rm eq}]}{2r(\tau)}\\
&\le\frac{(a-1)[r_z(0)+r_{\rm eq}]^2}{2r(\tau)}\le 0,
\end{align}
\end{subequations}
the monotonicity of $g_1(a)$ can be verified as
\begin{equation}
\frac{d}{da}g_1(a)=\frac{2}{1-r(\tau)^2}\frac{dr(\tau)}{da}\le 0.
\end{equation}
Taking the derivative of $g_2(a)$ with respect to $a$, we have
\begin{equation}
\frac{d}{da}g_2(a)=\frac{\mfr{c}-2r(0)r(\tau)}{2(1-a)^2[r_{\rm eq}+r_z(0)] r(\tau)},
\end{equation}
where, $\mfr{c}\coloneqq (a+1)\kappa+2r_z(0)[ar_z(0)+(a-1)r_{\rm eq}]$.
If $\mfr{c}\le 2r(0)r(\tau)$, then $dg_2(a)/da\le 0$ is immediately obtained.
If $\mfr{c}> 2r(0)r(\tau)$, then we have
\begin{subequations}
\begin{align}
\frac{d}{da}g_2(a)&=\frac{\mfr{c}^2-4r(0)^2r(\tau)^2}{2(1-a)^2r(\tau)[\mfr{c}+2r(0)r(\tau)][r_{\rm eq}+r_z(0)]}\\
&=\frac{\kappa(\kappa-4r_{\rm eq}[r_{\rm eq}+r_z(0)])}{2r(\tau)[\mfr{c}+2r(0)r(\tau)][r_{\rm eq}+r_z(0)]}.
\end{align}
\end{subequations}
Since $\kappa \le [r_z(0)+r_{\rm eq}][(1-a)r_{\rm eq}-(a+1)r_z(0)]$ and $r_{\rm eq}+r_z(0)\ge 0$, we can easily derive that $\kappa\le 4r_{\rm eq}[r_{\rm eq}+r_z(0)]$.
Consequently, we obtain $dg_2(a)/da\le 0$.

Next, from Eq.~\eqref{eq:kappa.cond}, the possible range of $a$ can be obtained as
\begin{equation}
0\le a\le \frac{r_{\rm eq}-r_z(0)}{r_{\rm eq}+r_z(0)}-\frac{\kappa}{[r_{\rm eq}+r_z(0)]^2}.
\end{equation}
Since $g(a)$ is a decreasing function, we obtain $g(a)\le g(0)$, which is equivalent to
\begin{subequations}
\begin{align}
\frac{r(\tau)-r(0)}{r_z(0)-r_z(\tau)}\ln\frac{1+r(\tau)}{1-r(\tau)}&\le\frac{r_{\rm eq}-r(0)}{r_z(0)+r_{\rm eq}}\ln\frac{1+r_{\rm eq}}{1-r_{\rm eq}}\\
&\le\ln\frac{1+r_{\rm eq}}{1-r_{\rm eq}}=\beta\omega.
\end{align}
\end{subequations}
Here, we have used $r_{\rm eq}-r(0)\le r_z(0)+r_{\rm eq}$.

Numerical verification of the bound is shown in Fig.~\ref{fig:num.result.two.lv.atom}.
\begin{figure}[t]
\centering
\includegraphics[width=0.8\linewidth]{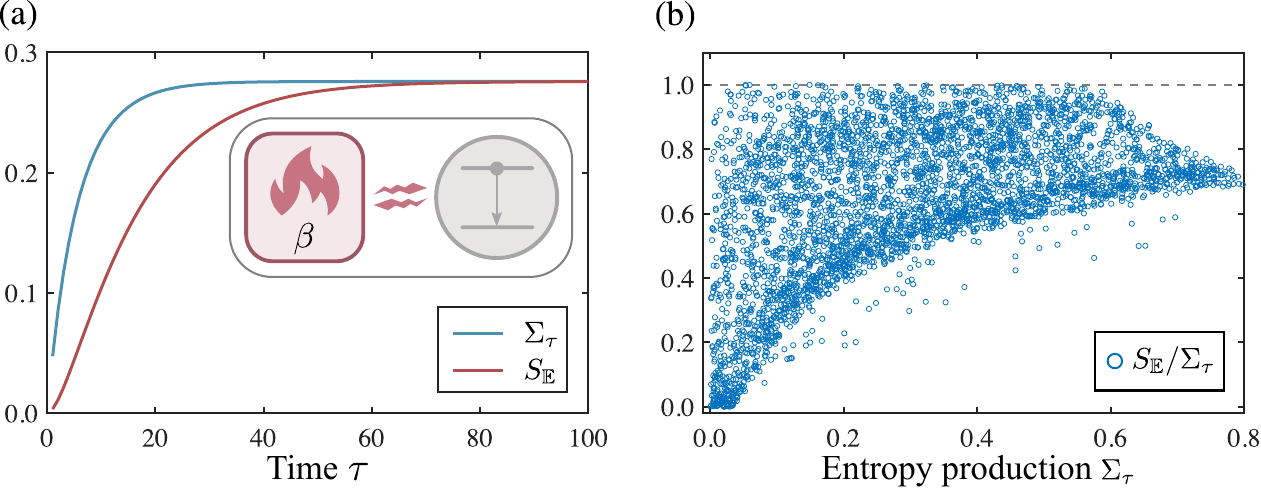}
\protect\caption{Numerical illustration of Eq.~(\EigenBound) in the two-level atom. (a) The irreversible entropy production $\Sigma_\tau$ and the lower-bound $S_{\mbb{E}}$ are plotted as functions of time $\tau$. Parameters are fixed as $\omega=0.5$, $\beta=0.5$, $\gamma=0.01$, and the initial density matrix is $\rho(0)=(\mbb{I}-0.8\sigma_z)/2$. (b) Random verification of the bound. The dashed line depicts the upper-bound of the ratio $S_{\mbb{E}}/\Sigma_\tau$. Each circle represents the ratio $S_{\mbb{E}}/\Sigma_\tau$ calculated with a random initial density matrix and $\tau\in[10^0,10^3]$.}\label{fig:num.result.two.lv.atom}
\end{figure}